\documentclass[11pt]{article}
\usepackage{url}
\usepackage{epsf}
\usepackage{graphicx}
\usepackage{amsfonts}
\usepackage{amssymb}
\usepackage{amsmath}
\usepackage{latexsym}
\usepackage{setspace, float}
\usepackage{color}
\usepackage[normalem]{ulem}
\usepackage[table]{xcolor}

\usepackage{makeidx}
\usepackage{multirow}
\usepackage{verbatim}
\usepackage{xy}\xyoption{all}

\newtheorem{theorem}{Theorem}

\newcommand{\qed}{\rule{0.5em}{1.5ex}}
\newcommand{\fqed}{{\hfill~\qed}}
\newenvironment{proof}{{\noindent \bf Proof.}}
                      {{\hfill \fqed} \vspace{1em}}
\newcommand\C{\mathcal{C}}

\newcommand{\mst}{${\sc MST}$}
\newcommand{\udg}{${\sc UDG}$}

\floatstyle{ruled}
\newfloat{algorithm}{htbp}{loa}
\floatname{algorithm}{Algorithm}

\title{Connectivity of Graphs Induced by Directional Antennas}
\author{
Mirela Damian\thanks{Villanova University, Villanova, USA, \texttt{mirela.damian@villanova.edu}.
\newline Partially supported by NSF grant CCF-0728909.}
\and
Robin Flatland\thanks{Siena College, Loudonville, USA, \texttt{flatland@siena.edu}.}}
\date{}

\begin{document}
\maketitle

\begin{abstract}
This paper addresses the problem of finding an orientation and a minimum radius
for directional antennas of a fixed angle placed at the points of a planar set
$S$, that induce a strongly connected communication graph.
We consider problem instances in which antenna angles are fixed at $90^\circ$
and $180^\circ$, and establish upper and lower bounds for the minimum radius
necessary to guarantee strong connectivity. In the case of $90^\circ$ angles, we
establish a lower bound of $2$ and an upper bound of $7$.
In the case of $180$ angles, we establish a lower bound of
$\sqrt{3}$ and an upper bound of $1+\sqrt{3}$.
Underlying our results is the assumption that the unit disk graph for
$S$ is connected.
\end{abstract}

\section{Introduction}
Let $S$ be a set of points in the plane representing wireless nodes.
Assume that each node is equipped with one directional antenna,
geometrically represented as a wedge with angular aperture $\alpha$ and radius $r$
(see Figure~\ref{fig:antenna}a). An antenna orientation is given by the
counterclockwise angle $\theta$ measured from the positive $x$-axis to the bisector
of the wedge.
The \emph{communication graph} $G(S)$ formed by the antennas placed at points
in $S$ is a directed graph with vertex set $S$ and edges $\overrightarrow{ab}$
directed from $a$ to $b$, if and only if the antenna wedge at $a$ covers $b$.
Let $\udg(S)$ denote the unit disk graph for $S$ (i.e., the graph in which any two
points in $S$ within unit distance are connected by an edge).
In this paper we address the following problem.

%\medskip
\begin{center}
\begin{minipage}{0.75\linewidth}
Let $S$ be a planar point set such that $\udg(S)$ is connected.
For a fixed angle $\alpha$, find an orientation $\theta$
of the antennas at the points in $S$ and a minimum radius $r$ for which
the communication graph $G(S)$ is \emph{strongly connected}.
\end{minipage}
\end{center}

\medskip
\noindent
We consider instances of this problem for $\alpha = 180^\circ$ (Section~\ref{sec:180}) and
$\alpha = 90^\circ$ (Section~\ref{sec:90}), and establish lower and upper bounds for the
minimum radius required to guarantee strong connectivity.
For the case $\alpha = 90^\circ$, we establish a lower bound of $2$ and an upper bound
of $7$.
For the case $\alpha = 180^\circ$ angles, we establish a lower bound of
$\sqrt{3}$ and an upper bound of $1+\sqrt{3}$. Underlying these
results is the assumption that $\udg(S)$ is connected. We note that the recent related work by 
Ben-Moshe et. al~\cite{bcckms-dawn-10}
also considers $90^\circ$-antennas but with orientations
restricted to the four standard quadrant directions, and it gives an algorithm for
constructing a bidirectional communication graph using a radius value 
$r$ that is at most twice the optimal value.

Throughout the paper, we use the notation $\mst(S)$ to refer to a minimum spanning tree of $S$
of \emph{maximum degree five}, which can be constructed using the algorithm
by Wu et al.~\cite{WDJLH06}. We say that a point $a \in S$ \emph{sees}
$b \in S$ if and only if the antenna wedge at $a$ covers $b$.

%%%%%%%%%%%%%%%%%%%%%%%%%%%%%%%%%Figure Begin
\begin{figure}[htpb]
\centering
\includegraphics[width=0.4\linewidth]{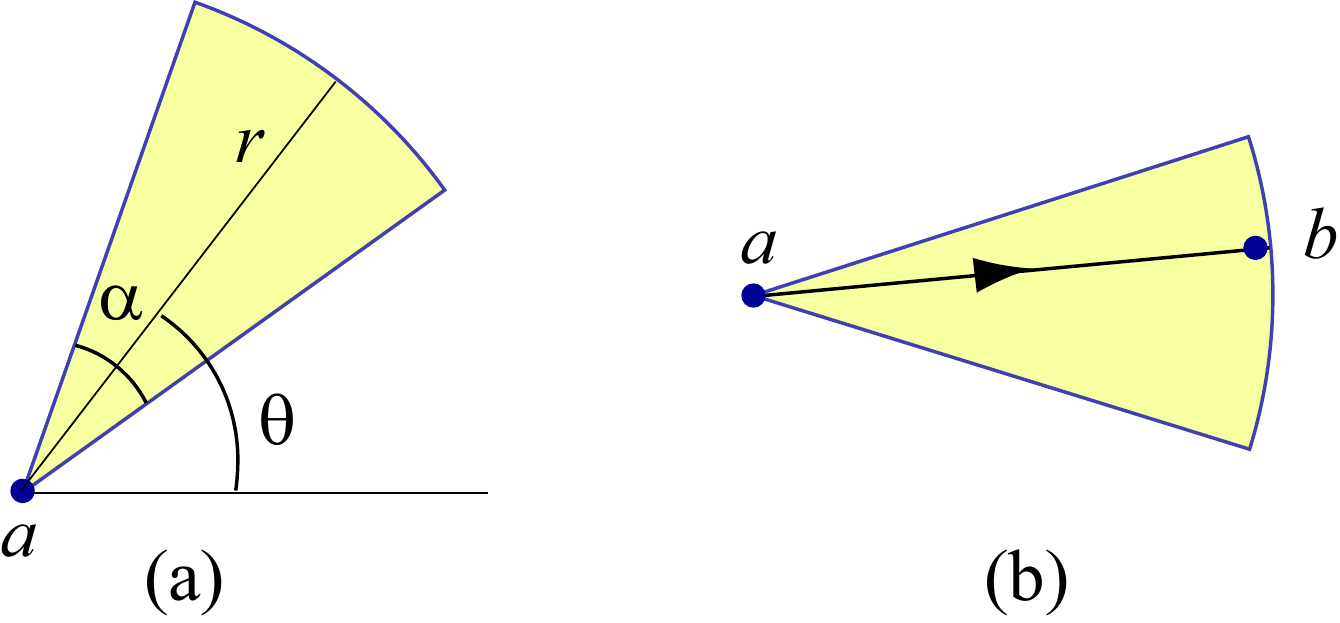}
\caption{(a) Directional antenna represented as a wedge of angle $\alpha$ and radius $r$ (b) $ab$ is a
directed edge in the communication graph.}
\label{fig:antenna}
\end{figure}
%%%%%%%%%%%%%%%%%%%%%%%%%%%%%%%%%Figure End

\section{$180^\circ$ Antennas}
\label{sec:180}

\begin{theorem}
For directional antennas with $\alpha = 180^\circ$, $r \geq \sqrt{3} $ is
sometimes necessary to build a strongly connected communication graph.
\label{thm:180main}
\end{theorem}
\begin{proof}
Figure~\ref{fig:lowerbound180} shows a point set for which
$r \geq \sqrt{3}$ is
necessary.
The solid line segments show the UDG; all angles are $120^\circ$.
Note that for any $r < \sqrt{3}$, any antenna placed at the
point labeled $p$ covers exactly two of its neighbors in the UDG
and no other points. Split the UDG into two pieces, $L$ and $R$,
by removing the edge
connecting $p$ to its uncovered neighbor. Let $R$ be the
piece containing $p$. Observe that
for any point $p' \neq p$ in $R$,
the distance from $p'$ to any point in $L$ is at least $\sqrt{3}$.
Since messages must flow from $R$ to $L$,
$r \geq \sqrt{3}$ is necessary.
\end{proof}
%%%%%%%%%%%%%%%%%%%%%%%%%%%%%%%%%Figure Begin
\begin{figure}[hp]
\centering
\includegraphics[width=0.65\linewidth]{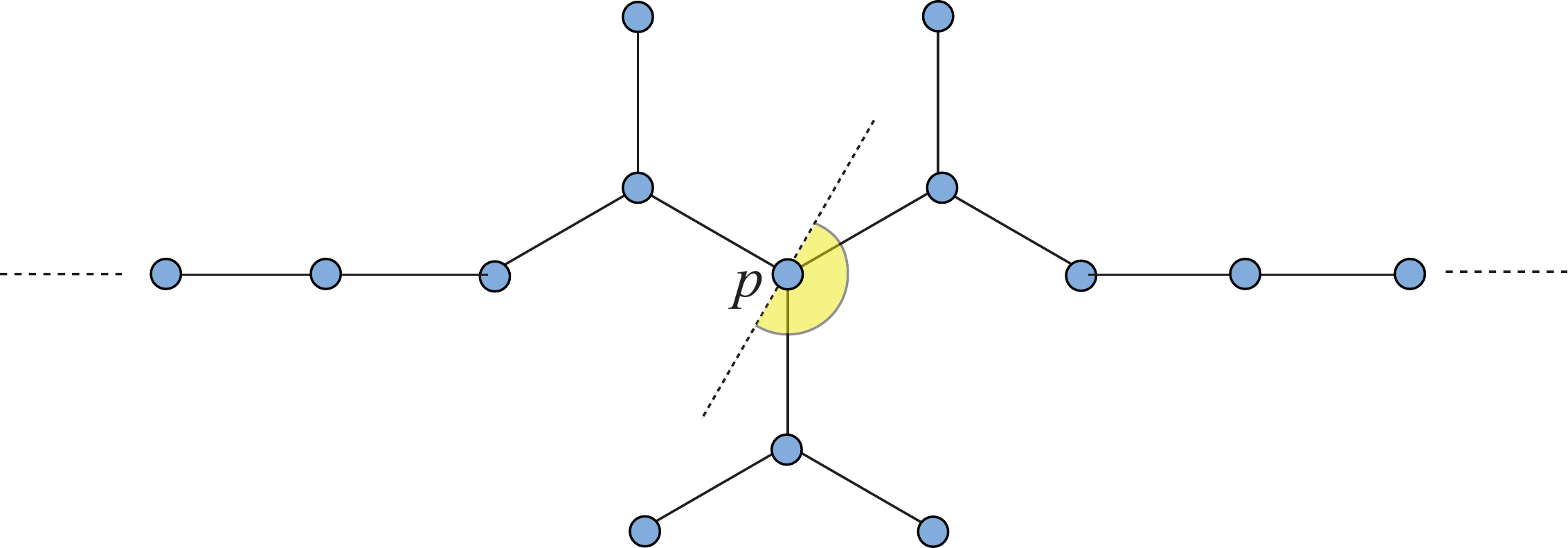}
\caption{$r \geq \sqrt{3}$ is necessary for this point set when $\alpha = 180^\circ$.}
\label{fig:lowerbound180}
\end{figure}
%%%%%%%%%%%%%%%%%%%%%%%%%%%%%%%%%Figure End

\begin{theorem}
For directional antennas with $\alpha = 180^\circ$, $r=1+\sqrt{3}$
suffices to build a strongly connected communication graph for a planar
point set $S$.
\end{theorem}
\begin{proof}
We begin by constructing $\mst(S)$.  Let a point in $S$ with a highest
$y$ coordinate be the root.
We first partition the nodes of $\mst(S)$ into groups, then show inductively how to
orient the antennas in each group to form the communication graph.
To identify the groups, pick a node of height one and place it in a group
along with its children.  Conceptually imagine removing this group from the
tree, then repeat the process until no nodes are left (with the
possible exception of the root).  See Figure~\ref{fig:groups} for an
example of node grouping.
%%%%%%%%%%%%%%%%%%%%%%%%%%%%%%%%%Figure Begin
\begin{figure}[htpb]
\centering
\includegraphics[width=0.35\linewidth]{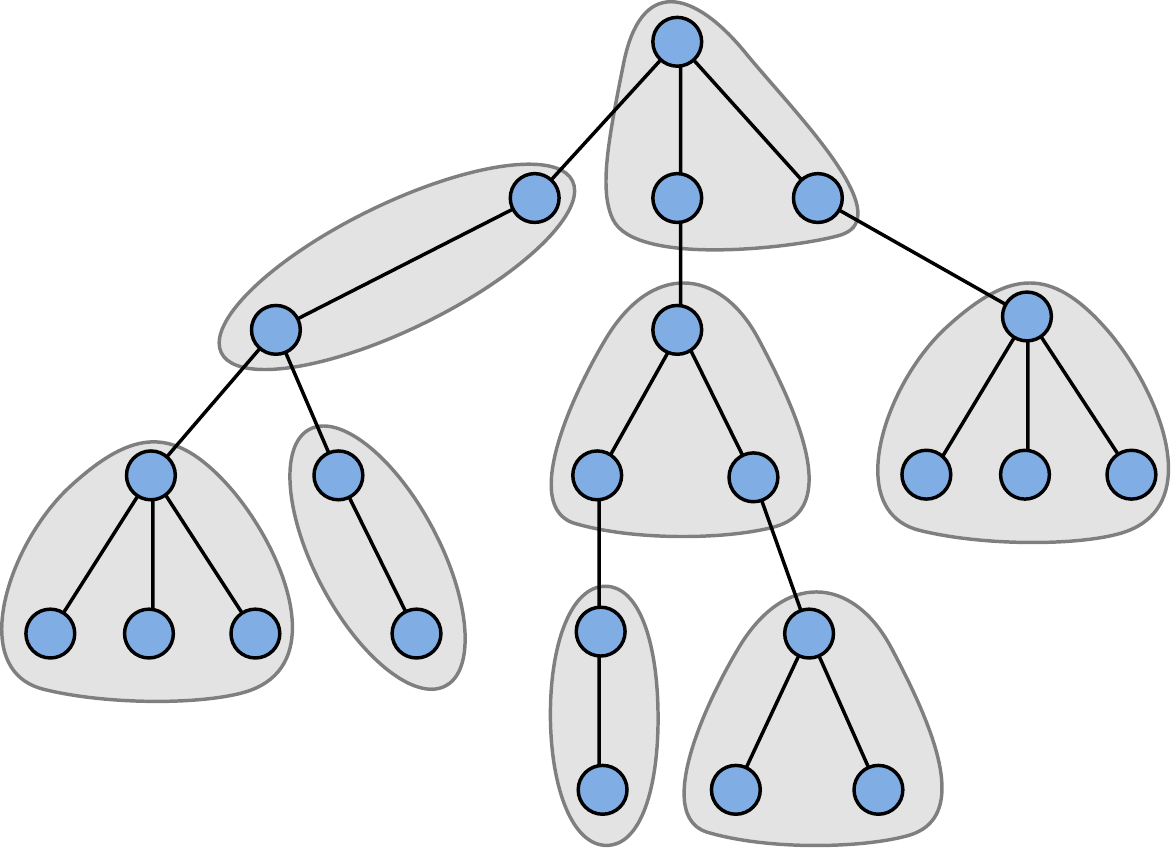}
\caption{Theorem~\ref{thm:180main}: Tree partitioned into groups.}
\label{fig:groups}
\end{figure}
%%%%%%%%%%%%%%%%%%%%%%%%%%%%%%%%%Figure End

We prove the theorem inductively on the number of groups $g$ in $\mst(S)$.
The base case corresponds to a tree with one group only ($g=1$). Let $p$ be
the parent node, and $c_1$ an arbitrary child of $p$.
Orient the antennas at $p$ and $c_1$ so that they are both aligned
with the segment $pc_1$ and cover opposite sides of the plane.
See Figure~\ref{fig:placement180basecase}.
This placement establishes direct bidirectional
communication between $p$ and $c_1$ since the two cones overlap along
the segment $pc_1$.
For the other children (if any), orient their antennas in any direction
that includes $p$. This enables direct communication from each child
$c$ to $p$.
Communication from $p$ to $c$ is indirect if $c$ lies outside the
antenna wedge at $p$, in which case the communication path is
$(p, c_1, c)$. Note that the distance from $c_1$ to any other child
of $p$ is at most $2$, therefore the radius $r > 2$ claimed by the
theorem guarantees connectivity from $c_1$ to $c$.

%%%%%%%%%%%%%%%%%%%%%%%%%%%%%%%%%Figure Begin
\begin{figure}[hptb]
\centering
\includegraphics[width=0.3\linewidth]{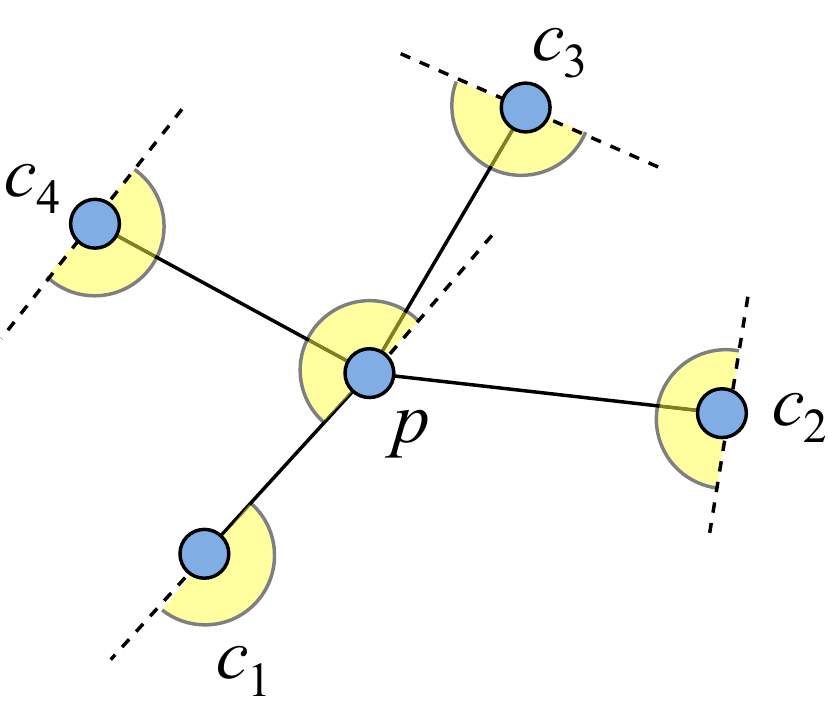}
\caption{Theorem~\ref{thm:180main}: Orientation of antennas in the base case.}
\label{fig:placement180basecase}
\end{figure}
%%%%%%%%%%%%%%%%%%%%%%%%%%%%%%%%%Figure End

The inductive hypothesis claims that, in the case of a tree
$\mst(S)$ composed of $g$ node groups, for some $g \ge 1$, there is an
orientation of antennas at the nodes of $\mst(S)$ that satisfies the theorem.
In addition, the inductive hypothesis requires that the root of $\mst(S)$
can reach any hop within a unit distance.
Note that this is true of the base case, with help from $c_1$
if the hop is not covered by $p$'s wedge.

To prove the inductive step, consider a tree $\mst(S)$ with $g+1$
groups. Let $p$ be the root of $\mst(S)$, and call the group containing
$p$ the \emph{root group}.
We discuss four cases, depending on the number of nodes in the root
group. The antenna placement for each of these cases is depicted in
Figure~\ref{fig:placement180}.
Observe that in the case of a triplet (Figure~\ref{fig:placement180}b),
$p$'s antenna is oriented so that is covers both children.
It can be verified, in the cases with one and two children depicted
in Figure~\ref{fig:placement180}(a, b),
$r \leq 2$ guarantees strong connectivity, and $p$ can reach any
hop within a unit distance.

The cases with the root group composed of four and five nodes
follow a similar pattern, once divided into pairs and triplets of nodes,
as depicted in Figure~\ref{fig:placement180} (c, d).
The dotted lines indicate children that have been paired.
Pairs consisting of two children
(see $c_2,c_3$ in Figure~\ref{fig:placement180}c) must be carefully
selected to achieve $r \leq 1 + \sqrt{3}$; the
requirement is that the two paired children form a smallest
angle at $p$. Since the cases under discussion involve
three or four children of $p$, an upper bound on a smallest such
angle is $120^\circ$. It follows that the distance between the paired
children is at most $\sqrt{3}$.
Then $r \leq 2$ guarantees the following:
(i) each pair and triplet of nodes is strongly connected, (ii)
each node in a pair can send messages to $p$ (because a node in a
pair can reach its counterpart node with
a setting $r = \sqrt{3}$, and at least one node in a pair can
reach $p$ with a setting $r = 1$), and
(iii) $p$ can reach any hop within unit distance (which includes its children).
It follows that the communication graph for the root group is strongly connected.
It remains to show that $\mst(S)$ is strongly connected.

(We pause here to note that a group cannot have five children, since each node
has degree at most five in $\mst(S)$, and the parent accounts for one of these
degree units. The one exception is the root of the entire tree. But since the
root is the point with a highest $y$ coordinate, all its children must lie in
a halfplane. The minimum angle separating two adjacent points in a planar
minimum spanning tree is $60^\circ$, so the root can also have most four
children.)

%%%%%%%%%%%%%%%%%%%%%%%%%%%%%%%%%Figure Begin
\begin{figure}[htpb]
\centering
\includegraphics[width=0.85\linewidth]{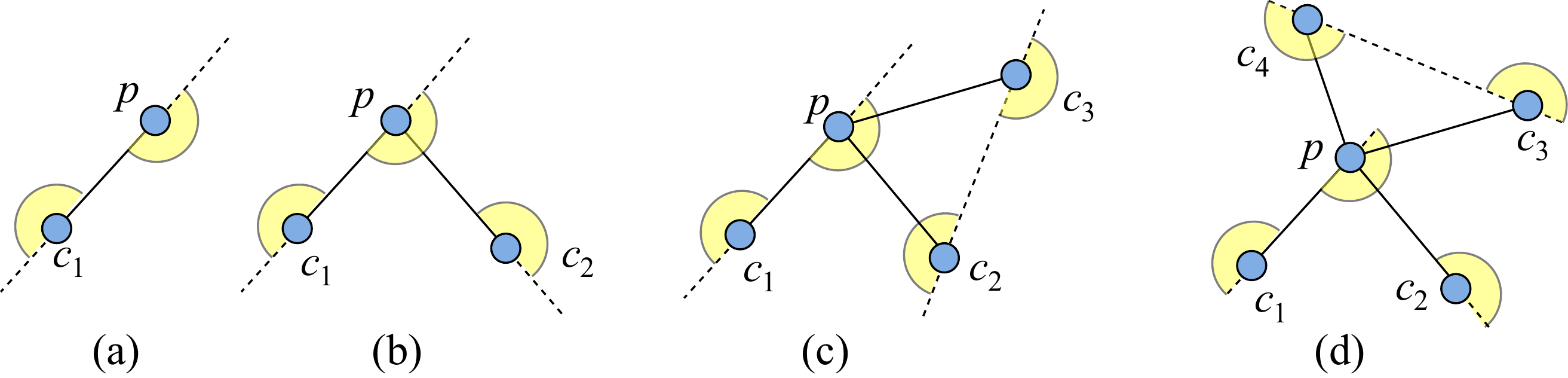}
\caption{Theorem~\ref{thm:180main}: Orientation of antennas in node groups.}
\label{fig:placement180}
\end{figure}
%%%%%%%%%%%%%%%%%%%%%%%%%%%%%%%%%Figure End

By the inductive hypothesis, each subtree $T \subset \mst(S)$
attached to a node in the root group forms a strongly connected
communication graph. In addition, the root of $T$ can reach any
hop within unit distance, and therefore it can send messages up to
the point of attachment. So to complete the inductive step, we
must show that each node in the root group can send messages down
to the root of an attached subtree. % with $r \leq 1 + \sqrt{3}$.
This is trivially true for $p$, since $p$ can communicate with any
point within unit distance, as established above.

Now consider the case when a child $c$ of $p$ is the attachment point
for a subtree with root $q$.  If $c$ plays the role of $c_1$ in a
pair or a triplet (as in Figures~\ref{fig:placement180}a,b) and
the antenna wedge at $c_1$ does not cover $q$, then $r \leq 2$
suffices to establish the communication path $(c_1, p, q)$.
If $c$ plays the role of $c_2$ in a triplet (as in Figure~\ref{fig:placement180}b)
and $c_2$ does not see $q$, then there are two cases to
consider. First, if $p$ sees $q$, then the communication path is
$(c_2, p, q)$. Otherwise, $c_1$ must see $q$.
The segment $c_2q$ cannot cross $c_1p$ since minimum spanning tree
edges do not cross.
This implies that $q$ is confined to the shaded region from Figure~\ref{fig:bound180},
therefore the distance from $q$ to $c_1$ is at most $2$. It follows that
$r \leq 2$ establishes the communication path $(c_2, p, c_1, q)$.

%%%%%%%%%%%%%%%%%%%%%%%%%%%%%%%%%Figure Begin
\begin{figure}[htpb]
\centering
\includegraphics[width=0.3\linewidth]{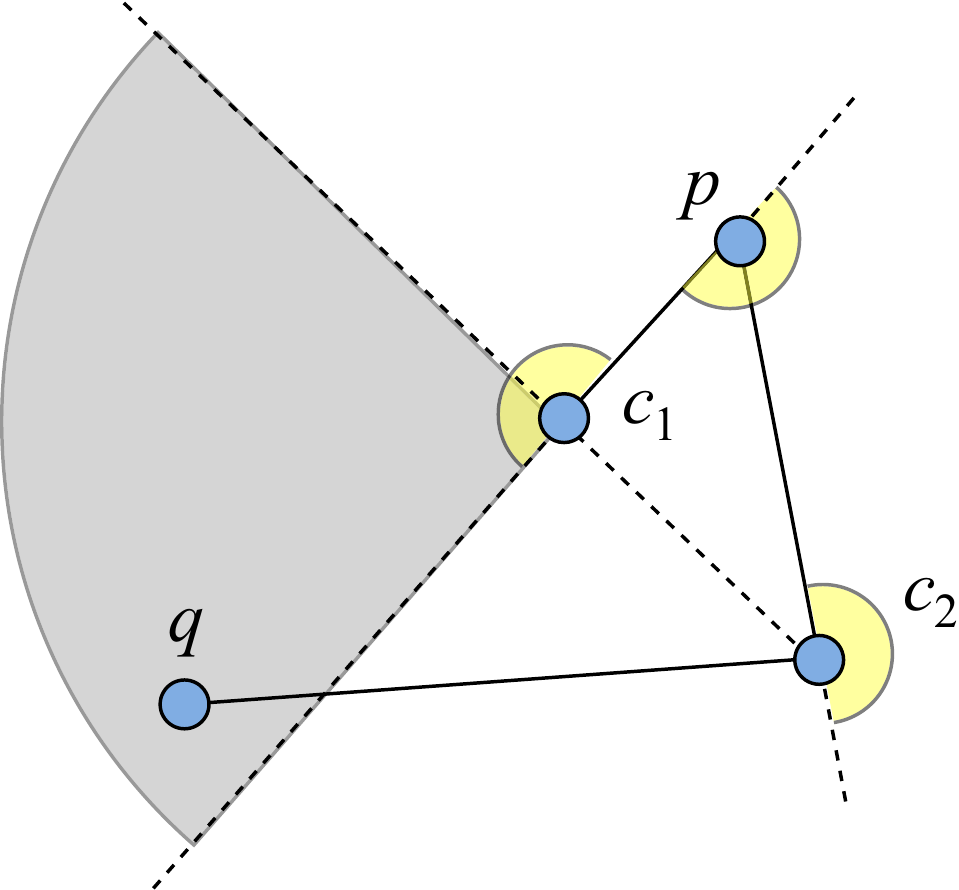}
\caption{The communication path from $c_2$ to $q$ is $(c_2, p, q)$.}
\label{fig:bound180}
\end{figure}
%%%%%%%%%%%%%%%%%%%%%%%%%%%%%%%%%Figure End

Finally, suppose that $q$ is attached to a child of $p$ that was
paired with another child of $p$, such as $c_2$ in
Figure~\ref{fig:placement180}c. If $c_2$ cannot see $q$,
then $c_3$ must be able to see $q$, so $c_2$ can use the communication path
$(c_2, c_3, q)$. Since the distance between $c_2$ and $c_3$
is at most $\sqrt{3}$, the last hop from $c_3$ to $q$ is no
longer than $1 + \sqrt{3}$, matching the upper bound on $r$
stated in the theorem.

It is possible that the root group contains a single node,
the root $p$ of $\mst(S)$. In this case, we deal
with $p$ separately by orienting its antenna in the negative $y$ direction.
The root is a highest point and therefore it can see all its
children, establishing direct communication with them.
By the inductive hypothesis, the children of $p$ can also send
messages to $p$, so this case is settled as well.
\end{proof}

\section{$90^\circ-$Antennas}
\label{sec:90}

\begin{theorem}
For directional antennas with $\alpha = 90^\circ$, $r \geq 2$ is
sometimes necessary to build a communication graph.
\end{theorem}
\begin{proof}
Consider a set of points positioned along the $x$-axis
at unit intervals.
An antenna placed at a point $p$ can only cover points to one side,
say its left side (so the antenna at $p$ is oriented to the left).
To enable messages to flow from points left of $p$ to points
right of $p$, the antenna at some point left of $p$ must be oriented
to the right. The nearest such point is $p$'s left neighbor, which
is at distance two from $p$'s right neighbor.
Therefore, $r = 2$ is necessary.
\end{proof}

\begin{theorem}
For any four points in general position, their $90^\circ-$antennas
can be oriented such that (i) a radius equal to the maximum pairwise
distance between the four points guarantees strong connectivity
of the four points, and (ii) the four antennas cover the entire plane.
\label{thm:fourpoints}
\end{theorem}
\begin{proof}
Consider first the case when the four points lie in convex
position. Let $a, b, c, d$ be the points in clockwise order around
the hull. Since $abcd$ is convex, the segments $ac$ and $bd$ must
intersect, as illustrated in Figure~\ref{fig:fourpoints}a.
Assume without loss of generality that $ac$ is the longer of the
two segments, and therefore the projection of at least one of $b$
and $d$ onto the line supported by $ac$ lies on the segment $ac$.
The orientation of antennas depends on the counterclockwise
angle $\beta$ from $ac$ to $bd$. We will assume $\beta \leq 90^\circ$;
the case when $90^\circ < \beta \leq 180^\circ$ is handled
symmetrically by reflection about the vertical.
It is not difficult to see that the orientation of antennas
from Figure~\ref{fig:fourpoints}a covers the entire plane, since
$ab$ and $cd$ intersect and $\beta$ is less than $90^\circ$. This
settles claim (ii) of the lemma. We now turn to claim (i) of the lemma.

%%%%%%%%%%%%%%%%%%%%%%%%%%%%%%%%%Figure Begin
\begin{figure}[htpb]
\centering
\includegraphics[width=\linewidth]{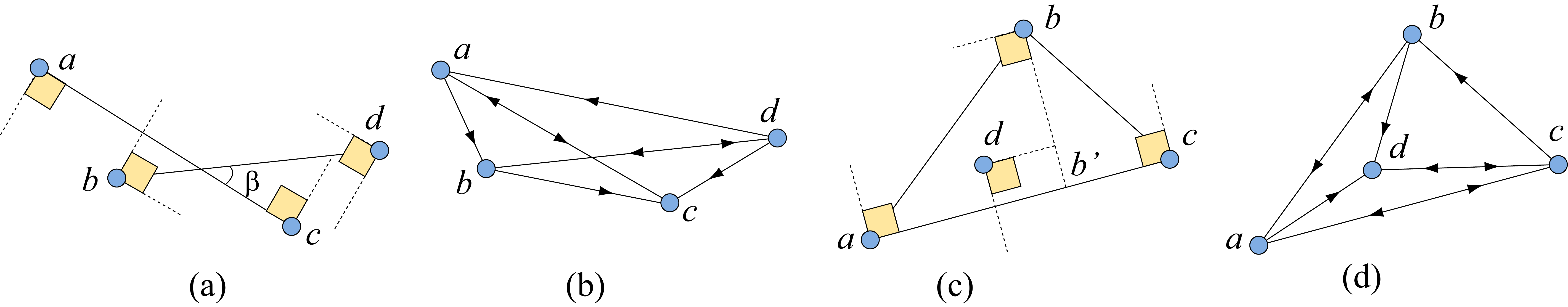}
\caption{ Antennas for points $a, b, c, d$ in convex position (a) and corresponding communication graph (b).
Antennas for points $a, b, c, d$ in non-convex position (c) and corresponding communication graph (d).}
\label{fig:fourpoints}
\end{figure}
%%%%%%%%%%%%%%%%%%%%%%%%%%%%%%%%%Figure End

Let each antenna wedge have radius equal to the maximum
pairwise distance between $a, b, c, d$.
First note that, for each node pair $(a, c)$ and $(b, d)$, each
node is contained in the antenna wedge of the counterpart node,
enabling direct communication between the nodes within a pair.
Communication between the pairs is settled as follows. Assume
that it is the projection of $d$ that lies on the segment
$ab$, as shown in Figure~\ref{fig:fourpoints}a.
Then $d$ is contained in $a$'s wedge, and $d$'s wedge
contains $c$, thus enabling full communication between the pairs, as
illustrated in Figure~\ref{fig:fourpoints}b.

Consider now the case when the four points do not lie
in convex position. Then three of the points, say $a, b, c$, comprise a
triangle that contains the fourth point, $d$. See Figure~\ref{fig:fourpoints}c.
Assume without loss of generality that $ac$ is a longest edge of
$\triangle abc$. Then the the projection $b'$ of $b$ onto $ac$
lies interior to the segment $ac$. Let $\triangle abb'$ contain $d$
(the situation when $\triangle cbb'$ contains $d$ is vertically
symmetric). The antenna orientation is depicted in
Figure~\ref{fig:fourpoints}c: all antenna wedges have one boundary line
parallel to $ac$; the antennas at points $a$ and $b$ face each other,
and similarly at points $c$ and $d$. This guarantees coverage of the entire
plane. In terms of connectivity, note that the nodes within each pair
$(a, b)$ and $(c, d)$ can communicate directly in both directions.
Since $d \in \triangle abb'$, both $a$ and $b$ can see $d$, and $c$ can
see $a$ and $b$ (recall that $ac$ is the longer side of $\triangle abc$,
therefore $\angle acb$ is acute, which implies that $c$ sees $b$).
This establishes full communication between the pairs.
\end{proof}

\begin{theorem}
For directional antennas with $\alpha = 90^\circ$, $r=7$ suffices
to build a strongly connected communication graph for a planar point set $S$.
\label{thm:90main}
\end{theorem}
\begin{proof}
The case when $S$ consists of two points $a$ and $b$ is trivial: orient the antennas
at $a$ and $b$ to point to each other.
If $S$ consists of three points $a, b$ and $c$, then $\triangle abc$ has at least two angles
strictly smaller than $90^\circ$. Orient the antennas at the apexes of these two angles
to cover the entire triangle, then orient the third antenna toward either of the other two
(see Figure~\ref{fig:3points}a).
%%%%%%%%%%%%%%%%%%%%%%%%%%%%%%%%%Figure Begin
\begin{figure}[hptb]
\centering
\includegraphics[width=0.4\linewidth]{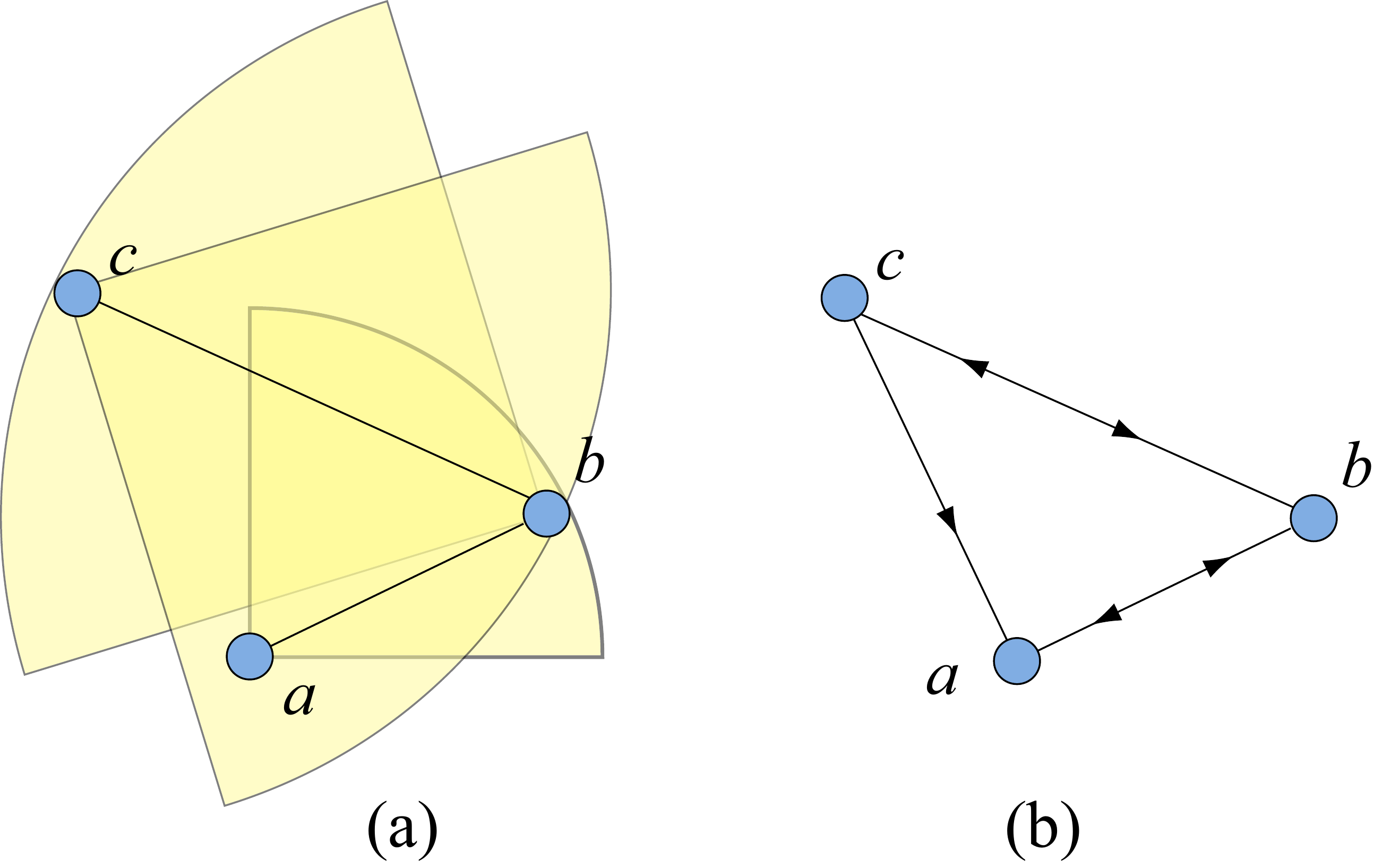}
\caption{A set $S$ of $3$ points. (a) Antenna orientation $\theta$ (b) Communication graph.}
\label{fig:3points}
\end{figure}
%%%%%%%%%%%%%%%%%%%%%%%%%%%%%%%%%Figure End
Then $r = 2$ suffices to form a strongly connected communication graph, since
$\max\{|ab|, |ac|, |bc|\} \le 2$.

We now turn to the general case $|S| \ge 4$.
We create groups of nodes in $\mst(S)$ recursively as follows.
Starting at the bottom of $\mst(S)$, identify a smallest subtree
$T \subseteq~\mst(S)$ of four or more nodes, whose removal does not disconnect
$\mst(S)$. It can be verified that such a subtree is topologically
equivalent to one of the subtrees shown in Figure~\ref{fig:nodegroups} (note
that the dashed connections are possible, but not required in the subtree).
Remove $T$ from $\mst(S)$ and recurse. The process stops when $\mst(S)$
is left with three or fewer nodes.
The result is a collection $\C$ of node groups, each with four or more vertices,
with the possible exception of the root subtree (the one containing
the root of $\mst(S)$). In each group we select four representative nodes, one
of which must be the root of the group subtree, and the other three could be
arbitrary.
For definitiveness we select the nodes shaded in Figure~\ref{fig:nodegroups} as
representative nodes.

%%%%%%%%%%%%%%%%%%%%%%%%%%%%%%%%%Figure Begin
\begin{figure}[hptb]
\centering
\includegraphics[width=0.55\linewidth]{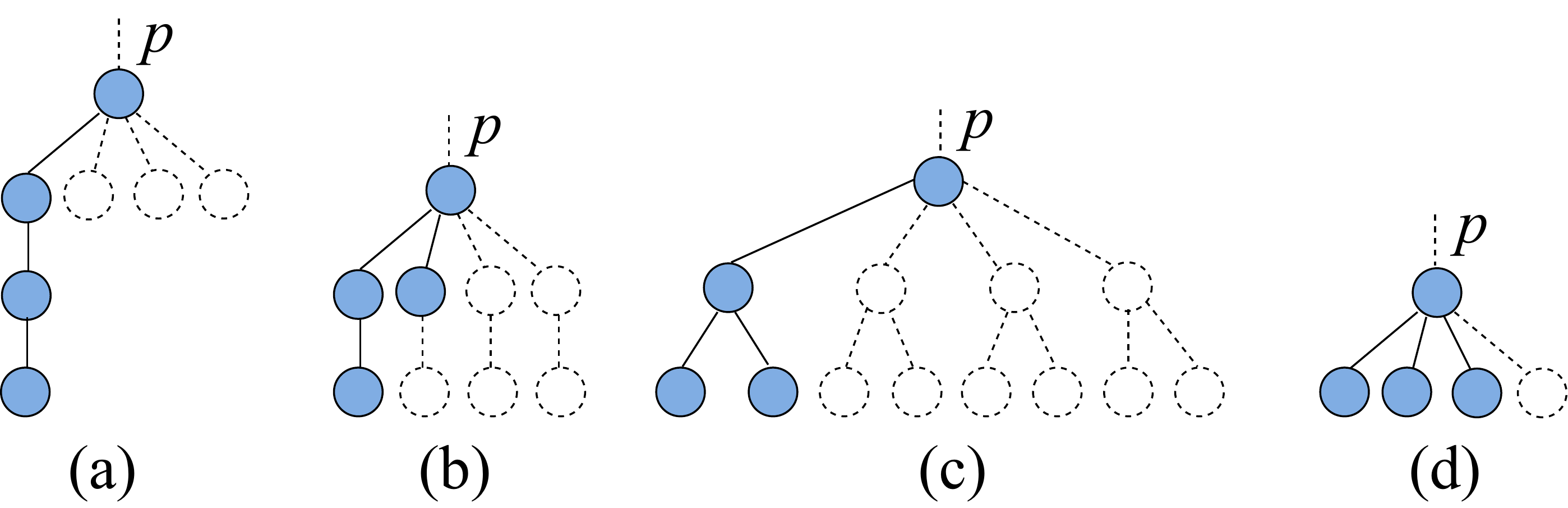}
\caption{Groups of four nodes that enable the use of Theorem~\ref{thm:fourpoints}. Dashed connections may or may not exist.}
\label{fig:nodegroups}
\end{figure}
%%%%%%%%%%%%%%%%%%%%%%%%%%%%%%%%%Figure End

% each node group g \in \C induces a subgraph H(g) \subseteq \mst(S)
% process subtrees of \mst(S) that are unions of groups in \C
% proof by induction on the number of groups in these subtrees
% four representative nodes in each group; root must be in; the other 3 are arbitrary
% Starting at the bottom of $\mst(S)$, identify a smallest subtree of four or more nodes,
% whose removal does not disconnect $\mst(S)$.
%

We prove the theorem inductively on the number of groups $g$ in $\mst(S)$.
The base case with $g = 1$ corresponds to a group of nodes arranged in a subtree
topologically equivalent to one of the trees depicted in Figure~\ref{fig:nodegroups}.
The representative node set in each group is $R = \{p, a, b, c\}$, with $p$ the
root of the group subtree. Note that the maximum pairwise distance between
nodes in $R$ is $d_{\max} = 3$ for the subtree depicted in
Figure~\ref{fig:nodegroups}a, and $d_{\max} = 2$ for the subtrees
depicted in Figure~\ref{fig:nodegroups}(b, c, d).
We use Theorem~\ref{thm:fourpoints} on $R$ to determine an orientation of the
antennas at nodes in $R$ that strongly connects $R$, for $r = d_{max}$.
Then $r = d_{max}+1$ enables any node in $R$ to reach any hop
within unit distance, because the antennas at nodes in $R$
collectively cover the entire plane.

The inductive hypothesis is that there is an orientation of antennas at the nodes
of $\mst(S)$ consisting of $g$ or fewer groups, that satisfies the theorem.
In addition, the inductive hypothesis requires that the root of $\mst(S)$
can reach any hop within a unit distance. This additional requirement was already
established for the base case.

To prove the inductive step, consider a tree $\mst(S)$ with $g+1$ groups. Assume first
that the root group contains at least four nodes, so they are arranged in a subtree
$T \subseteq \mst(S)$ topologically equivalent to
one of the trees from Figure~\ref{fig:nodegroups}.
As in the base case, we orient the antennas at the representative nodes of
$T$ as in Theorem~\ref{thm:fourpoints}, to establish coverage of the plane and
strong connectivity between these nodes, for $r = d_{max}$.
For each non-representative node $x$, orient the antenna at $x$ to point towards
a closest representative node $y$.
A simple analysis of the tree topologies from Figure~\ref{fig:nodegroups}
indicate that, in order to establish a connection from $x$ to $y$, a
radius of $1$ for the antenna at $x$ suffices for the cases depicted in
Figure~\ref{fig:nodegroups}(a,d), and a radius of $2$ suffices for the cases depicted
in Figure~\ref{fig:nodegroups}(b, c). Summing up these values with $d_{max}$,
we obtain that $r = 4$ establishes full connectivity between the nodes of $T$
(since one of the nodes in $R$ can reach $x$ in this case as well).
We now show that $r = 5$ guarantee strong connectivity of $\mst(S)$.

%%%%%%%%%%%%%%%%%%%%%%%%%%%%%%%%%Figure Begin
\begin{figure}[hptb]
\centering
\includegraphics[width=\linewidth]{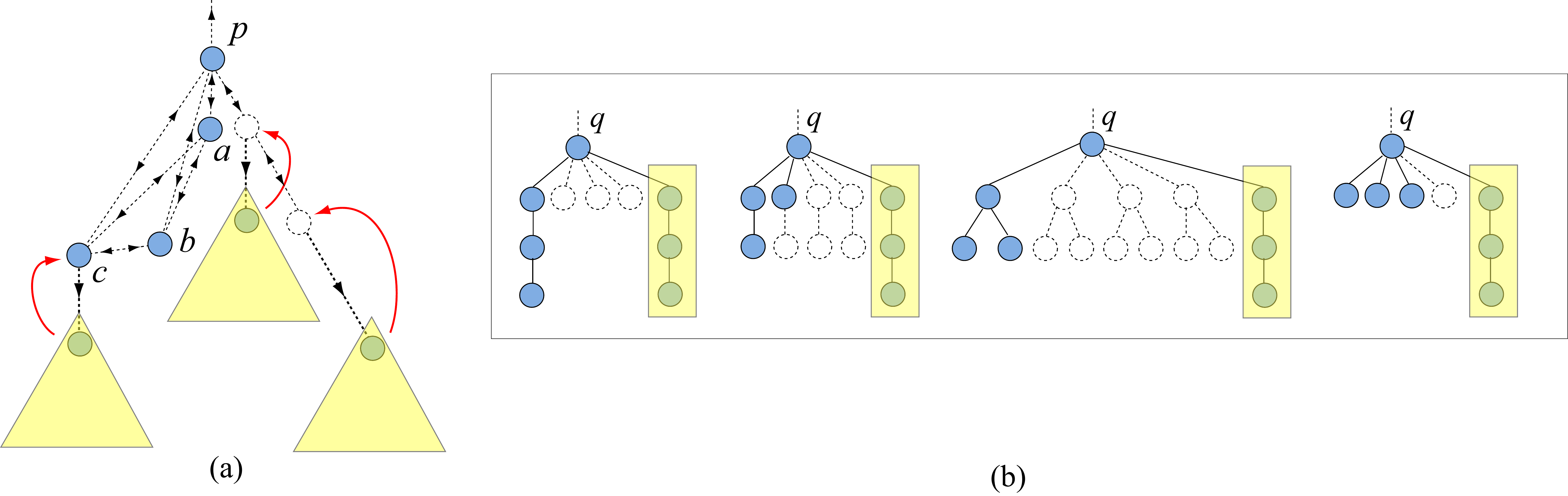}
\caption{(a) The inductive step for Theorem~\ref{thm:90main}. The communication paths indicated by circular arcs are guaranteed by the inductive hypothesis. The directed dashed edges represent communication paths established in the inductive step. (b) The root subtree with three nodes (boxed) viewed as part of a child subtree.}
\label{fig:90induction}
\end{figure}
%%%%%%%%%%%%%%%%%%%%%%%%%%%%%%%%%Figure End

The inductive hypothesis, along with the fact that each child in $\mst(S)$ is
within unit distance from its parent, implies that each subtree attached to a node
$x \in T$ can send messages up to $x$
(see the circular arcs in Figure~\ref{fig:90induction}a).
We have established that a setting of $r = 4$ enables strong connectivity
between the nodes of $T$. It follows that $r = 5$ enables each node
$x \in T$ to reach each child $y$ of $x$, because with this setting at
least one of the nodes in $R$ can reach $y$ (since their antennas cover the entire plane),
and $x$ can reach any node in $R$.
In addition, $r = 5$ enables the root of $\mst(S)$ to reach any hop
within unit distance (by a similar argument).
This settles the inductive step for this case.

If the node group at the root of $\mst(S)$ contains fewer than four nodes,
this group can be viewed as attached to the root $q$ of any adjacent node group
in $\mst(S)$. This idea is illustrated in Figure~\ref{fig:90induction}b, where
a root subtree of three nodes is attached in turn to each ``full'' subtree (with four or more nodes).
Regardless of the topological structure of the root subtree $T$, the
maximum distance between any two nodes in $T$ does not exceed $2$.
Orient the antennas at each node in $T$ towards $q$. Then
each node in $T$ can reach $q$ with $r = 3$. A simple analysis
of the configurations from Figure~\ref{fig:90induction}b shows that
a representative node of the subtree rooted at $q$ can reach any node in $T$
with an increase of $2$ in its transmission radius. Then the same analysis
as before shows that $r = 7$ settles the inductive step. It is likely that
a more complex analysis of this special case (with the root group composed of
three or fewer nodes) can maintain the previously established bound of $r=5$.
Such an analysis is left for future work.
\end{proof}

%\newpage
%\small
%\bibliographystyle{plain}
%\bibliography{antenna}

\end{document}